\documentclass{amsart}
\usepackage{amssymb}
\usepackage[all]{xy}

\theoremstyle{plain}
\newtheorem{theorem}{Theorem}[section]
\newtheorem{lemma}[theorem]{Lemma}
\newtheorem{example}{Example}
\newtheorem{remark}{Remark}

\newtheorem{corollary}[theorem]{Corollary}
\newtheorem{proposition}[theorem]{Proposition}
\newcommand{\nck}[2]{{#1 \choose #2}}
\newcommand{\id}[1]{\langle#1\rangle}
\newcommand{\R}{\frac{GR(p^a,m)[x]}{\id{x^{p^s}+1}}}
\newcommand{\T}{\frac{GR(2^a,m)[x]}{\id{x^{2^s}-1}}}
\newcommand{\Tm}{\mathcal{T}_m}

\def\cent{60}
\def\ra{110}
\def\rb{100}
\def\rc{80}
\def\rd{60}
\def\re{40}
\def\rf{20}
\def\rg{10}
\def\spa{0}
\def\spb{40}
\def\spc{80}
\def\spd{120}
\newcommand{\figa}
{
\label{fig.one}

\xy  \POS (\cent,\ra)
*{\id{1}}="a1", (\cent,\rb) *{\id{3,x+1}}="b1", (\spa,\rc)
*{\id{3,(x+1)^2}}="c1", (\spb,\rc) *{\id{3+(x+1)}}="c2", (\spc,\rc)
*{\id{3+2(x+1)}}="c3", (\spd,\rc) *{\id{x+1}}="c4", (\spa,\rd)
*{\id{3+(x+1)^2,3(x+1)}}="d1", (\spb,\rd) *{\id{3}}="d2", (\spc,\rd)
*{\id{3+2(x+1)^2}}="d3", (\spd,\rd) *{\id{(x+1)^2}}="d4", (\spa,\re) *{\id{3+6(x+1)+(x+1)^2}}="e1", (\spb,\re) *{\id{3+3(x+1)+(x+1)^2}}="e2", (\spc,\re) *{\id{3+(x+1)^2}}="e3", (\spd,\re)
*{\id{3(x+1)}}="e4", (\cent,\rf) *{\id{3(x+1)^2}}="f1", (\cent,\rg)
*{\id{0}}="g1"

\POS"a1" \ar @{-} "b1"

\POS"b1" \ar @{-} "c1"

\POS"b1" \ar @{-} "c2"

\POS"b1" \ar @{-} "c3"

\POS"b1" \ar @{-} "c4"

\POS"c1" \ar @{-} "d1"

\POS"c1" \ar @{-} "d2"

\POS"c1" \ar @{-} "d3"

\POS"c1" \ar @{-} "d4"

\POS"c2" \ar @{-} "d4"

\POS"c3" \ar @{-} "d4"

\POS"c4" \ar @{-} "d4"

\POS"d1" \ar @{-} "e1"

\POS"d1" \ar @{-} "e2"

\POS"d1" \ar @{-} "e3"

\POS"d1" \ar @{-} "e4"

\POS"d2" \ar @{-} "e4"

\POS"d3" \ar @{-} "e4"

\POS"d4" \ar @{-} "e4"

\POS"e1" \ar @{-} "f1"

\POS"e2" \ar @{-} "f1"

\POS"e3" \ar @{-} "f1"

\POS"e4" \ar @{-} "f1"

\POS"f1" \ar @{-} "g1"
\endxy
}

\numberwithin{equation}{section}

\title [On the Hamming weight of Repeated Root Codes]{On the Hamming weight of Repeated Root Cyclic and Negacyclic Codes over Galois Rings}
\author{Sergio R. L\'opez-Permouth}
\address{Department of Mathematics, Ohio University, Athens,
Ohio-45701, USA}
\author{Steve Szabo}
\address{Department of Mathematics, Ohio University, Athens,
Ohio-45701, USA}

\date{}

\begin{document}
\maketitle
\begin{center}

\end{center}

\begin{abstract}
Repeated root Cyclic and Negacyclic codes over Galois rings have
been studied much less than their simple root counterparts. This
situation is beginning to change. For example, repeated root codes
of length $p^s$, where $p$ is the characteristic of the alphabet
ring, have been studied under some additional hypotheses. In each
one of those cases, the ambient space for the codes has turned out
to be a chain ring. In this paper, all remaining cases of cyclic and
negacyclic codes of length $p^s$ over a Galois ring alphabet are
considered. In these cases the ambient space is a local ring with
simple socle but not a chain ring. Nonetheless, by reducing the
problem to one dealing with uniserial subambients, a method for
computing the Hamming distance of these codes is provided.
\end{abstract}

\section{Introduction}
Cyclic and negacyclic codes have been studied extensively in many
contexts, beginning with their linear versions over finite fields
and continuing on to the study of such codes over a finite ring
alphabet $A$. A common element in the study of these codes is that
they are precisely the submodules of the free module $A^n$ that
correspond to the ideals of a suitable ring $R_n$ which is
isomorphic to $A^n$ as an $A$-module.  The ring $R_n$ is either the
quotient $\frac{A[x]}{\id{x^n-1}}$ (for the cyclic case) or the
quotient ring $\frac{A[x]}{\id{x^n+1}}$ (for the negacyclic case).
In either case, we refer to the ring $R_n$ as \textit{the ambient
space} or \textit{ambient ring} for the codes. While the literature
on cyclic and negacylic codes over chain rings (such as Galois
rings) has grown in leaps and bounds (see
\cite{calderbank_1995,kanwar_1997,kiah_2008,pless_1996,salagean_2006,wolfmann_1999}),
in most instances the studies have been focused only on the cases
where the characteristic of the alphabet ring is coprime to the code
length, the so-called simple root codes. A few of the contributions
to the study of the cases where the characteristic of the alphabet
ring is not coprime the the code length (repeated root codes) are
\cite{abualrub_2003,abualrub_2003_2,castagnoli_1991,dougherty_2007,
ozbudak, vanlint_1991}. In this paper we focus on the repeated root
case where the code length is in fact a power of a prime.

Let $p$ be a prime and consider cyclic and negacyclic codes of
length $p^s$ over $GR(p^a,m)$. The study of such codes was started
in \cite{dinh_2004} for the negacyclic case when $p=2$ and $m=1$. It
was shown there that the ambient $\frac{Z_{2^a}[x]}{\id{x^{p^s}+1}}$
is a chain ring. This result was extended to the case when $m$ is
arbitrary in \cite{dinh_2005}. The distances for most of these codes
was calculated there. The chain ring structure of the ambient was
heavily used to accomplish this goal.

When $a=1$, the Galois ring $GR(p^a,m)$ is just the Galois field
$F_{p^m}$. Codes over $F_{p^m}$ were consider in \cite{dinh_2008}.
There it was shown that for arbitrary $p$, the ambient space
$\frac{F_{p^m}[x]}{\id{x^{p^s}+1}}$ is a chain ring. Once again the
chain structure of the ambient space was used to compute all the
code distances. Then it was shown also in \cite{dinh_2008} that
$\frac{F_{p^m}[x]}{\id{x^{p^s}+1}}\cong
\frac{F_{p^m}[x]}{\id{x^{p^s}-1}}$ when $p$ is odd, which allows all
the negacyclic results to be carried over to the cyclic code case.
It should be noted that over a field of characteristic 2, there is
no distinction between cyclic and negacyclic codes since
$\frac{F_{2^m}[x]}{\id{x^{p^s}+1}}=
\frac{F_{2^m}[x]}{\id{x^{p^s}-1}}$.

In all cases mentioned so far, the codes correspond to principal
ideals. This is a consequence of the fact that the code ambients are
chain rings. In the remaining cases, which comprise the primary
subject of this paper, the code ambients are no longer chain rings
and in fact, not even PIRs. There are three remaining cases:
negacyclic codes over $GR(p^a,m)$ for odd prime $p$ and $a>1$ of
length $p^s$; cyclic codes of the same type; cyclic codes over
$GR(2^a,m)$ for $a>1$ of length $p^s$. In this paper these remaining
cases are considered and a method for computing the Hamming distance
of any code is provided.

Now, simple root cyclic codes over $Z_{p^m}$ were studied in
\cite{calderbank_1995} where a generating set for such codes was
formulated and it was also proved that these codes are principal
ideals of the ambient ring. An alternative generating set was given
for codes over $Z_4$ in \cite{pless_1996}. This result was extended
to $Z_{p^m}$ in \cite{kanwar_1997} where they also showed the
connection between the two formulations. These results were in turn
extended to simple root cyclic codes over Galois rings in
\cite{wan_1999}.

In a series of papers
(\cite{salagean_2006},\cite{norton_2001},\cite{norton_2003},\cite{norton_2000}),
the idea of Gr\"obner basis was extended to principal ideal rings
and was used to prove the existence of generating sets with certain
desirable properties for cyclic and negacyclic codes over chain
rings. Specifically, they showed that given this generating set, the
code distance can be determined from one particular element in the
generating set. In Section \ref{sect.struct}, we will use this
theory to determine all minimum code distances.

For the most part, the literature preceding
(\cite{salagean_2006},\cite{norton_2001},\cite{norton_2003},\cite{norton_2000}),
failed to address specific distance information about cyclic and
negacyclic codes. The generating sets given in
\cite{calderbank_1995} are based on the factorization of $x^n-1$.
Given a factorization of $x^n-1$, it is still not simple to compute
distances, even in light of the results in \cite{salagean_2006}
mentioned earlier. To use those results in this context, the minimum
weights for all principal ideals are needed. Since it was shown that
simple root cyclic codes over Galois rings are principal, the
results just mentioned bring us no closer to finding distances in
the simple root case. In this paper however, we will show that these
results can be very useful in determining distances in some multiple
root codes where not all of the codes are principal. This method
reduces the problem to finding distance information of related codes
which are principal.

In Section \ref{sect.back}, the necessary background on Galois rings
is given together with other results  that are needed throughout the
paper. Section \ref{sect.struct} considers the class of codes in
$\R$ where $p$ is an odd prime and $a>1$. In this section it is
shown that $\R$ is a local ring with simple socle that is not a
chain ring. Then a method for computing Hamming distances is shown.
Section \ref{sect.struct2} examines cyclic codes. When $p$ is odd,
there is a one-one correspondence between cyclic and negacylcic
codes over $GR(p^a,m)$ of length $p^s$ for odd prime $p$ which is
shown. The remainder of the section is devoted to cyclic codes over
$GR(2^a,m)$ for $a>1$ of length $p^s$. It is shown that $\T$ has a
very similar structure to $\R$ from Section \ref{sect.struct}. Again
a method for computing Hamming distances is shown.

\section{Preliminaries}
\label{sect.back} In this paper, the word \textit{ring} means finite
commutative ring with identity. The only exception is when we talk
about the (infinite) ring $R[x]$ of polynomials with coefficients in
the ring $R$. A \textit{local ring} is a ring with a unique maximal
ideal. Given a commutative ring $R$, the \textit{Jacobson radical}
of $R$, denoted by $J(R)$, is the intersection of all maximal ideals
of $R$ and the \textit{socle} of $R$, denoted by $soc(R)$, is the
sum of all minimal ideals of $R$. A polynomial $f(x)\in R[x]$ is
\textit{regular} if it is not a zero divisor. The following is a
characterization of regular polynomials in polynomial rings over
local rings.

\begin{lemma}[Theorem XIII.2, \cite{mcdonald_1974}]
Let $R$ be a finite local commutative ring and $f(x)\in R[x]$ where
$f(x)=a_0+\dots+a_nx^n$ for $a_i\in R$. The following are
equivalent:
\begin{enumerate}
\item $f$ is a regular polynomial.
\item $\id{a_0,\dots,a_n}=R$.
\item $a_i$ is a unit for some $0\leq i\leq n$.
\item $f(x)\pmod{p}\neq 0$.
\end{enumerate}
\end{lemma}

Polynomial rings over local rings admit a division algorithm for
certain polynomials.

\begin{lemma}[Proposition 3.4.4, \cite{bini_2002}]
\label{lemma.div} Let $R$ be a finite local commutative ring and
$f(x),g(x)\in R[x]$ where $g(x)$ is regular. Then there exists
$q(x),r(x)\in R[x]$ such that
\[
f(x)=g(x)q(x)+r(x)
\]
with $deg(r)<deg(g)$ or $r(x)=0$.
\end{lemma}

A \textit{chain ring} is a ring whose ideals are linearly ordered by
inclusion. The following characterization of chain rings is well-known:

\begin{lemma}
\label{lemma.chain} Let $R$ be a finite commutative ring. The
following are equivalent:
\begin{enumerate}
\item $R$ is a chain ring.
\item $R$ is a local principal ideal ring.
\item $R$ is a local ring with maximal ideal that is principal.
\end{enumerate}
\end{lemma}

Galois rings constitute a very important family of finite chain
rings. They can be defined as follows: Let $f(x)\in Z_{p^a}[x]$ be a
basic irreducible polynomial (a \textit{basic irreducible}
polynomial in $Z_{p^m}[x]$ is an irreducible polynomial in
$Z_{p^a}[x]$ whose reduction modulo $p$ is irreducible in
$Z_{p}[x]$) and $m=deg(f)$. Then the Galois ring
$GR(p^a,m)=\frac{Z_{p^a}[x]}{\id{f(x)}}$. It is well-known that
different choices of $m$ and $a$ yield non-isomorphic Galois rings
while, on the other hand, distinct choices of $f(x)$ with the same
degree $m$ yield the same Galois ring up to isomorphism. We now list
a few pertinent details about these rings. For a more detailed
account of the theory of Galois rings including proofs of the
results we mention here, see \cite{mcdonald_1974} or
\cite{wan_1999}.

Every Galois ring $R=GR(p^a,m)$ contains a $(p^m-1)^{th}$ primitive
root of unity $\zeta$. Every $r\in R$ has a $p$-adic expansion
$r=\zeta_0+\zeta_1p+\dots+\zeta_{a-1}p^{a-1}$ where $\zeta_i\in
\{0,1, \zeta,\zeta^2,\dots,\zeta^{p^m-2}\}$, the Teichm\"uller set
$\Tm$ of $R$.

Given a polynomial $f(x)$ in any polynomial ring $R[x]$, $f$ can be
viewed in the form $f(x)=\sum_{i=0}^{k}a_i(x+1)^i$ where $a_i\in R$.
So, for $f\in GR(p^a,m)[x]$, $f(x)=\sum_{i=0}^k \sum_{j=0}^a
\zeta_{ij}p^j(x+1)^i$ where $\zeta_{ij}\in \{0,1,
\zeta,\zeta^2,\dots,\zeta^{p^m-2}\}$.

The next two Lemmas are results on negacyclic code ambients over
Galois rings which will be needed in the proceeding sections.
Defining multiplication of $r\in\R$ by
$m\in\frac{GR(p,m)[x]}{\id{x^{p^s}+1}}$ as multiplication in $\R$
mod $p$, $\frac{GR(p,m)[x]}{\id{x^{p^s}+1}}$ can be made into an
$\R$-module. In light of this, the following lemma is easy to see.

\begin{lemma}
\label{lemma.iso} For any prime $p$, the $\R$-modules $p^{a-1}\R$
and $\frac{GR(p,m)[x]}{\id{x^{p^s}+1}}$ are isomorphic.
\end{lemma}

\begin{lemma}[\cite{dinh_2008}, Proposition 3.2]
\label{lemma.hai} For any prime $p$, the ambient ring
$\frac{GR(p,m)[x]}{\id{x^{p^s}+1}}$ is a chain ring with exactly the
following deals,
\[
\frac{GR(p,m)[x]}{\id{x^{p^s}+1}}=\id{(x+1)^0}\supsetneq\dots\supsetneq\id{(x+1)^{p^s}}=0.
\]
\end{lemma}

In \cite{norton_2001} an algorithm was given to find a Gr\"obner
basis for ideals of a polynomial ring over a PIR. Later in
\cite{salagean_2006}, it is shown that any ideal of a residue ring
of a polynomial ring over a chain ring has a Gr\"obner basis with
certain additional properties. Since for any prime $p$, $GR(p^a,m)$
is a chain ring, ideals of $\R$ will have such a Gr\"obner basis.
The following Lemma is a restatement of that result.

\begin{lemma}[adapted from Theorem 4.1 in \cite{salagean_2006}]
\label{lemma.ideal} For any prime $p$, given an ideal
$I\vartriangleleft\R$, for $i\in\{0,\dots,r\}$ there exist $j_i\in
Z$ and $f_i\in \R$ where ${0\leq r \leq a-1}$ such that
\[
I=\id{p^{j_0}f_0,\dots,p^{j_r}f_{r}}
\]
 and
\begin{enumerate}
\item $0\leq j_0<\dots<j_{r}\leq a-1$
\item $f_i$ monic for $i=0,\dots,r$,
\item $p^s>deg(f_0)>\dots>deg(f_r)$,
\item $p^{j_{i+1}}f_i\in\id{p^{j_{i+1}}f_{i+1},\dots,p^{j_{r}}f_{r}}$
\item $p^{j_0}(x^{p^s}+1)\in\id{p^{j_0}f_0,\dots,p^{j_r}f_{r}}$ in
$GR(p^a,m)[x]$.
\end{enumerate}
\end{lemma}

One can show further that the set of generators in Lemma
\ref{lemma.ideal} is a strong Gr\"obner basis in the sense of
\cite{salagean_2006}. While interesting, this fact will not be used
here.

\begin{lemma}
\label{lemma.pd} Let p be a prime. Let $k\leq\frac{p^n}{2}$ and $l$
be the largest integer s.t. $p^l\mid k$. Then
$p^{n-l}\mid\nck{p^n}{k}$.
\end{lemma}
\begin{proof}
For $k\leq p$, the result holds. Now we proceed in 3 cases. First
assume there is an $l>0$ s.t. $p^l\mid k-1$ and it is the largest
such integer. Then $p^{n-l}\mid\nck{p^n}{k-1}$. Since $p^l\mid k-1$,
$p\nmid k$ and $p^l\mid p^n-k+1$. So, $p^l\mid\frac{p^n-k+1}{k}$.
Hence, $p^{n-l+l}\mid\nck{p^n}{k-1}\frac{p^n-k+1}{k}=\nck{p^n}{k}$.
Now, assume $p\nmid k-1$ and $p\nmid k$. Then $p\nmid p^n-k+1$. So,
for any $l$ s.t. $p^l\mid\nck{p^n}{k-1}$, $p^l\mid\nck{p^n}{k}$.
Noting the previous case, $p^n\mid\nck{p^n}{k}$. Now, assume there
is an $l>0$ s.t. $p^l\mid k$ and it is the largest such integer.
Then $p\nmid k-1$ and so $p\nmid p^n-k+1$. Again noting the previous
cases, $p^{n-l}\mid\nck{p^n}{k}$
\end{proof}

\section{Negacyclic codes in $\R$ for odd prime $p$}
\label{sect.struct} As mentioned earlier, all ambient rings
previously studied in the literature are chain rings. They are
$\frac{GR(2^a,m)[x]}{\id{x^{2^s}+1}}$,
$\frac{GR(p,m)[x]}{\id{x^{p^s}+1}}$ and
$\frac{GR(p,m)[x]}{\id{x^{p^s}-1}}$ for $a,m,p,s\in Z$ where $a\geq
1$, $m\geq 1$, $p$ is prime and $s\geq 0$. In the following sections
the remaining cases will be studied. In these remaining cases, the
ambient spaces are not chain rings. We will show, however, that they
are local rings with simple socle.

In this section, the structure of $\R$ where $p$ is an odd prime and
$a>1$ is studied so that in the following section the structure
details can be used to find Hamming distance of all codes. Since
$s=0$ is the trivial case also assume $s>0$.

We start by showing that $x+1$ is nilpotent. The calculation of its exact
nilpotency is saved for Corollary \ref{cor.nil}.

\begin{proposition}
\label{prop.nilp} In $\R$, $(x+1)$ is nilpotent.
\end{proposition}

\begin{proof}
\[
\begin{array}{lll}
(x+1)^{p^s}     & = & x^{p^s}+\nck{p^s}{p^s-1}x^{p^s-1}+\dots+\nck{p^s}{1}x+1\\
                & = & x^{p^s}+1+p\alpha(x)\\
                & = & p\alpha(x)\\
\end{array}
\]
where $\alpha(x)\in \frac{GR(p^a,m)[x]}{\id{x^{p^s}+1}}$. Then
$(x+1)^{p^sa}=p^a(\alpha(x))^a=0$.
\end{proof}

\begin{proposition}
\label{prop.max} The ambient ring $\R$ is local with radical
$J\left(\R\right)=\id{p,x+1}$.
\end{proposition}

\begin{proof}
Let $I$ be the set of non-invertible elements.
Let $f\in\id{p,x+1}$. By Proposition \ref{prop.nilp} $(x+1)$ is
nilpotent. Since $p$ is also, $f$ is nilpotent and hence not
invertible. So,
$\id{p,x+1}\subset I$. Now, let $f\in I$. We can write
$f(x)=\sum_{i=0}^{k}a_i(x+1)^i$ where $a_i\in GR(p^a,m)$. So, $f$ is
invertible if and only if $a_0$ is invertible. The $p$-adic expansion
$a_0=\sum_{i=0}^{a-1}b_ip^i$ where
$b_i\in\{0,1,\zeta,\zeta^2,\dots,\zeta^{p^{m-2}}\}$ assures that $a_0$ is
invertible if and only if $b_0\neq 0$. The assumption that $f$ not invertible implies
therefore, that $p\mid a_0$ and this shows that $f\in \id{p,x+1}$.
So, $I\subset\id{p,x+1}$. Since $I$ contains all invertible
elements, it is the unique maximal ideal and therefore, $\R$ is
local.
\end{proof}

\begin{proposition}
\label{prop.soc} The socle $soc\left(\R\right)$ of $\R$
is the simple module $\id{p^{a-1}(x+1)^{(p^s-1)}}$.
\end{proposition}
\begin{proof}
Using Lemma \ref{lemma.hai}, it can be shown that
$\id{p^{a-1}(x+1)^{(p^s-1)}}\subset soc\left(\R\right)$. Let
$a(x)\in soc\left(\R\right)$. Since $J(\R)=\id{p,x+1}$, $pa(x)=0$
and $(x+1)a(x)=0$ so $a(x)\in \id{p^{a-1}(x+1)^{p^s-1}}$. Hence,
$soc\left(\R\right)= \id{p^{a-1}(x+1)^{(p^s-1)}}$.

By Lemmas \ref{lemma.iso} and \ref{lemma.hai}, it is clear that
$soc\left(\R\right)$ is simple.
\end{proof}

\begin{proposition}
\label{lemma.notchain} In the ambient ring $\frac{GR(p^a,m)[x]}{\id{x^{p^s}+1}}$
\begin{enumerate}
\item $p\notin{\id{x+1}}$
\item $x+1\notin{\id{p}}$
\item $\R$ is not a chain ring
\item $\id{p,x+1}$ is not a principal ideal
\end{enumerate}
\end{proposition}
\begin{proof} Assume $p\in\id{x+1}$. So, $p=(x+1)f(x)+(x^{p^s}+1)g(x)$ in
$GR(p^a,m)[x]$. When $x=-1$, $p=0$ which is a contradiction in this
case since $a\neq 1$. Hence, $p\notin{\id{x+1}}$.

Now assume $x+1\in{\id{p}}$. Then $x+1=pf(x)+(x^{p^s}+1)g(x)$. Now,
comparing coefficients, $1=pf_0+g_0$ which implies $1\equiv g_0$
modulo $p$. Also, $0=pf_{p^s}+g_0+g_{p^s}$ which implies $g_0\equiv
-g_{p^s}$ modulo $p$. In general,
$0=pf_{kp^s}+g_{(k-1)p^s}+g_{kp^s}$ for $k\geq 1$. So, $g_{kp^s}\neq
0$ for $k\geq 0$ which is a contradiction since $g$ is a polynomial.
Hence, $x+1\notin{\id{p}}$.

Finally, since $0\neq{\id{p}}\nsubseteq{\id{x+1}}$ and
$0\neq{\id{x+1}}\nsubseteq{\id{p}}$, $\R$ is not a chain ring. Since
any local ring with principal maximal ideal is a chain ring by Lemma
\ref{lemma.chain}, $\R$ cannot have principal maximal ideal. Hence,
$\id{p,x+1}$ is 2-generated.
\end{proof}

\begin{theorem}
\label{theo.notc} The ambient ring $\R$ is a finite local ring with
simple socle but not a chain ring.
\end{theorem}
\begin{proof}
Result of Propositions \ref{prop.max}, \ref{prop.soc} and
\ref{lemma.notchain}.
\end{proof}

\begin{example} To illustrate Theorem \ref{theo.notc}, we provide the following figure.
It shows the ideal lattice of $\frac{Z_{3^2}[x]}{\id{x^{3}+1}}$.
Notice that the radical is $\id{3,x+1}$ and the socle is
$\id{3(x+1)^2}$. More importantly, we see that the ring is not a
chain ring. \bigskip

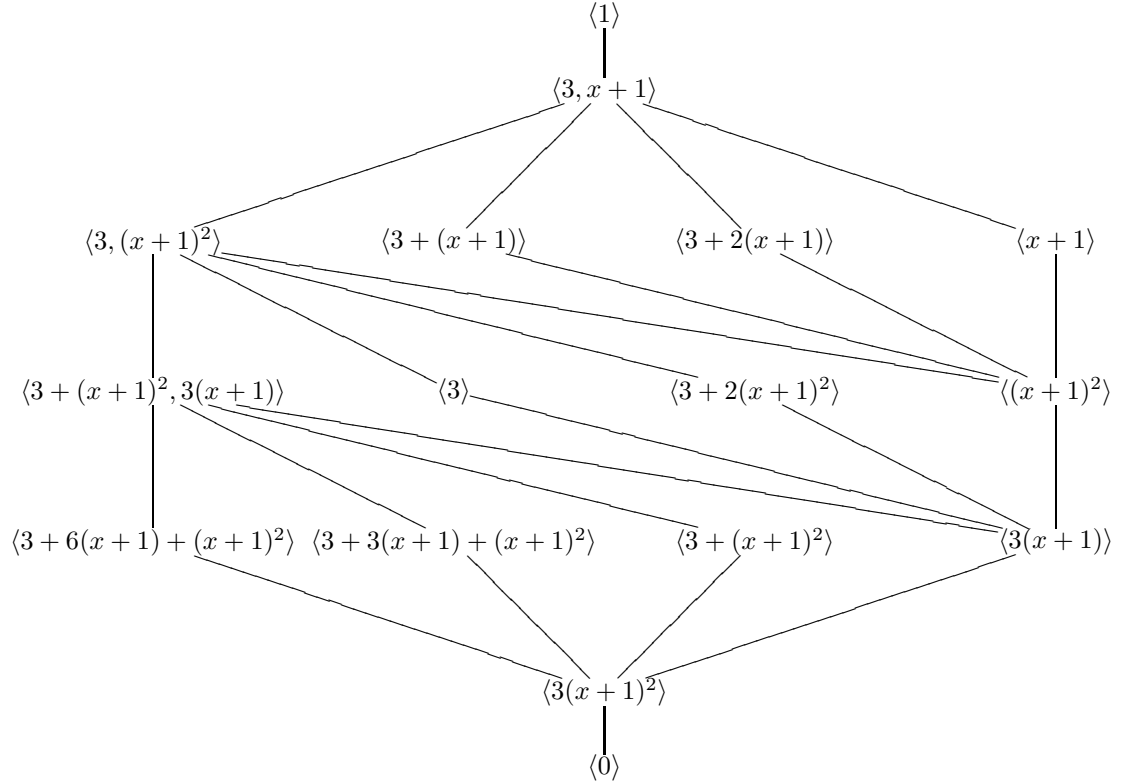
\figa
\end{example}

Now we are ready to develop our main structural Lemma.

\begin{lemma}
\label{lemma.main} In $\frac{GR(p^a,m)[x]}{\id{x^{p^s}+1}}$ for
$t\geq 0$,
\[
(x+1)^{p^s+t(p-1)p^{s-1}}=p^{t+1}b_t(x)(x+1)^{p^{s-1}}+a_t(x)
\]
where $b_t(x)$ is invertible and $p^{t+2}|a_t(x)$.
\end{lemma}

\begin{proof} We proceed by induction on $t$. For $t=0$,
\[
\begin{array}{lcl}
0=x^{p^s}+1&=&((x+1)-1)^{p^s}+1\\
&=&(x+1)^{p^s}-\nck{p^s}{{p^s}-1}(x+1)^{{p^s}-1}+\nck{{p^s}}{{p^s}-2}(x+1)^{{p^s}-2}-\dots+\nck{{p^s}}{1}(x+1)\\
\end{array}
\]
By Lemma \ref{lemma.pd}
\[
\begin{array}{lcl}
(x+1)^{p^s}&=&\nck{p^s}{{p^s}-1}(x+1)^{{p^s}-1}-\nck{{p^s}}{{p^s}-2}(x+1)^{{p^s}-2}+\dots-\nck{{p^s}}{1}(x+1)\\
&=&\nck{p^s}{(p-1)p^{s-1}}(x+1)^{(p-1)p^{s-1}}+\dots-\nck{p^s}{p^{s-1}}(x+1)^{p^{s-1}}+a_0(x)\\
&=&pb_0(x)(x+1)^{p^{s-1}}+a_0(x)\\
\end{array}
\]
for some $a_0(x)$ s.t. $p^2|a_0(x)$ and $b_0(x)$ invertible.

Now assume the result holds for $t-1$. So there exists some
$a_{t-1}(x)$ s.t. $p^{t+1}|a_{t-1}(x)$ and $b_{t-1}(x)$ invertible
where
$(x+1)^{p^s+(t-1)(p-1)p^{s-1}}=p^tb_{t-1}(x)(x+1)^{p^{s-1}}+a_{t-1}(x)$.
So
\[
\begin{array}{lcl}
(x+1)^{p^s+t(p-1)p^{s-1}}&=&(x+1)^{p^s+(t-1)(p-1)p^{s-1}}(x+1)^{(p-1)p^{s-1}}\\
&=&\left[p^tb_{t-1}(x)(x+1)^{p^{s-1}}+a_{t-1}(x)\right](x+1)^{(p-1)p^{s-1}}\\
&=&p^tb_{t-1}(x)(x+1)^{p^s}+a_{t-1}(x)(x+1)^{(p-1)p^{s-1}}\\
&=&p^tb_{t-1}(x)\left[pb_0(x)(x+1)^{p^{s-1}}+a_0(x)\right]+a_{t-1}(x)(x+1)^{(p-1)p^{s-1}}\\
&=&p^{t+1}b_{t-1}(x)b_0(x)(x+1)^{p^{s-1}}+p^tb_{t-1}(x)a_0(x)+a_{t-1}(x)(x+1)^{(p-1)p^{s-1}}\\
&=&p^{t+1}b_{t-1}(x)\left[b_0(x)+\frac{a_{t-1}(x)}{p^{t+1}}(x+1)^{(p-2)p^{s-1}}\right](x+1)^{p^{s-1}}+p^tb_{t-1}(x)a_0(x)\\
&=&p^{t+1}b_{t}(x)(x+1)^{p^{s-1}}+a_t(x)\\
\end{array}
\]
\end{proof}

\begin{corollary}
\label{cor.nil} In $\R$, the nilpotency of $x+1$ is
$p^sa-p^{s-1}(a-1)$.
\end{corollary}
\begin{proof}
By Lemma \ref{lemma.main},
\[
(x+1)^{p^s+(a-2)(p-1)p^{s-1}}=p^{a-1}b(x)(x+1)^{p^{s-1}}+a(x)
\]
for some $b(x)$ is invertible and $a(x)$ s.t. $p^{a}|a(x)$. So,
$a(x)=0$ and
\[
(x+1)^{p^s+(a-2)(p-1)p^{s-1}}=p^{a-1}b(x)(x+1)^{p^{s-1}}.
\]
So,
\[
(x+1)^{p^s+(a-2)(p-1)p^{s-1}}(x+1)^{(p-1)p^{s-1}-1}=p^{a-1}b(x)(x+1)^{p^s-1}
\]
meaning
\[
(x+1)^{p^s+(a-1)(p-1)p^{s-1}-1}=p^{a-1}b(x)(x+1)^{p^s-1}\neq 0.
\]
Finally,
\[
(x+1)^{p^s+(a-1)(p-1)p^{s-1}}=p^{a-1}b(x)(x+1)^{p^s}=0
\]
Hence the nilpotency of $x+1$ is
$p^s+(a-1)(p-1)p^{s-1}=p^sa-p^{s-1}(a-1)$.
\end{proof}

So far we have seen that $\R$ is not a chain ring and not even a
PIR. Although a description of a generating set for ideals would be
most desirable, we will settle for a bound on the number of
generators. We provide two proofs of this result. The first one has
a more theoretical flavor as it uses results from
\cite{mcdonald_1974} on polynomial rings over local rings. The
second one aims at establishing a simple algorithm for producing
such a generating set.

\begin{lemma}
\label{lemma.agen} In the ambient ring
$\frac{GR(p^a,m)[x]}{\id{x^{p^s}+1}}$, any ideal is generated by $a$
or fewer elements.
\end{lemma}
\begin{proof}[Proof \#1]
Let $I=\id{f_0,f_1,\dots,f_n}\subset GR(p^a,m)[x]$ where $f_i\in
GR(p^a,m)[x]$ and $deg(f_i)<p^s$. There exists a regular polynomial
$g_i\in GR(p^a,m)[x]$ where $f_i=p^{k_i}g_i$. Consider $f_a\neq f_b$
where $deg(f_a)\geq deg(f_b)$. If $k_a=k_b$, by using Lemma
\ref{lemma.div} on $g_a$ and $g_b$, $f_a\in \id{f_b,r}$ where
$r(x)=0$ or $deg(r)<deg(f_b)$. So, at this point
$I=\id{f_0,f_1,\dots,f_{a-1},r,f_{a+1},\dots,f_n}$. If
$p^{k_b+1}\nmid r(x)$, we can continue this process and replace
$f_j$ and then $r$ etc. until the remainder is divisible by
$p^{k_j+1}$. It is clear then using this process that it will
produce a generating set $I=\id{g_0,p^1g_1,\dots,p^{a-1}g_{a-1}}$
where either each $g_i=0$ or is a regular polynomial.
\end{proof}

For the second proof, a canonical form for the description of
polynomials is needed. It was shown earlier that any $f\in\R$ can be
written as
\[
f(x)=\sum_{i=0}^{p^s-1} \sum_{j=0}^{a-1} \zeta_{ij}p^j(x+1)^i
\]
where $\zeta_{ij}\in \{0,1,
\zeta,\zeta^2,\dots,\zeta^{p^m-2}\}=\Tm$. Now, we will show there is
a useful form which we call the canonical form. First, we can write
\[
f(x)=\beta_0p^{0}(x+1)^{i_0}\alpha_0(x)+\beta_1p^{1}(x+1)^{i_1}\alpha_1(x)+\dots+\beta_{a-1}p^{{a-1}}(x+1)^{i_{a-1}}\alpha_{a-1}(x)
\]
where $\beta_k\in \Tm$ and $\alpha_k(x)\in\R$ is invertible. Assume
for a moment that $i_{k_1}\leq i_{k_2}$ and $\beta_{k_1}\neq 0 \neq
\beta_{k_2}$ for some $k_1\neq k_2$. Then
\[
\beta_{k_1}p^{k_1}(x+1)^{i_{k_1}}\alpha_{k_1}+\beta_{k_2}p^{{k_2}}(x+1)^{i_{k_2}}\alpha_{k_2}=\beta_{k_1}p^{{k_1}}(x+1)^{i_{k_1}}(\alpha_{k_1}+\frac{\beta_{k_2}}{\beta_{k_1}}p^{{k_2}-{k_1}}(x+1)^{i_{k_2}-i_{k_2}}\alpha_{k_2}).
\]
Since
$\alpha_{k_1}+\frac{\beta_{k_2}}{\beta_{k_1}}p^{{k_2}-{k_1}}(x+1)^{i_{k_2}-i_{k_2}}\alpha_{k_2}$
is invertible, wlog we can assume $i_0>i_1>\dots>i_{a-1}$. We now
use this canonical form in the following proof.

\begin{proof}[Proof \#2]
Let $I=\id{g_0,g_1,\dots,g_n}$ where $g_j\in\R$. If $n\leq a-1$ we
are done so assume $n\geq a$. Let $f_j=g_j$. Viewing the $f_j$ in
canonical form, write
\[
f_j(x)=\beta_{j0}p^0(x+1)^{i_{j0}}\alpha_{j0}(x)+\dots+\beta_{j(a-1)}p^{a-1}(x+1)^{i_{j(a-1)}}\alpha_{j(a-1)}(x).
\]
If $\beta_{j0}\neq0$ for some $j$, we reorder the $f_j$ so that for
all $j$ where $\beta_{j0}\neq0$, $i_{00}\leq i_{j0}$. Then, let
\[
f_j^\prime(x)=f_j-\beta_{00}^{-1}\alpha_{00}^{-1}f_0(x)\beta_{j0}\alpha_{j0}(x)(x+1)^{i_{j0}-i_{00}}.
\]
If $\beta_{j0}=0$ for all $j$, let $f_j^\prime(x)=f_j$. In either
case, $f_j^\prime\in I$ and
$f_j\in\id{f_0,f_1^\prime,\dots,f_n^\prime}$. Note that
$\beta_{j0}^\prime=0$ for $j\geq1$. To avoid unnecessary
complication with the notation at this point we let
$f_j=f_j^\prime$. Next we do the same process but we leave $f_0$
alone. If $\beta_{j1}\neq0$ for some $j\geq1$, we reorder the $f_j$
so that for all $j\geq 1$ where $\beta_{j1}\neq0$, $i_{11}\leq
i_{j1}$. Then, let
$f_j^{\prime}(x)=f_j-\beta_{11}^{-1}{\alpha_{11}}^{-1}f_1(x)\beta_{j1}\alpha_{j1}(x)(x+1)^{i_{j1}-i_{11}}$.
If $\beta_{j1}=0$ for all $j\geq1$, let $f_j^{\prime}(x)=f_j$. In
either case, $f_j^{\prime}\in \id{f_0,f_1,\dots,f_n}$ and
$f_j\in\id{f_0,f_1,f_2^{\prime}\dots,f_n^{\prime}}$. Note that
$\beta_{j1}^\prime=0$ for $j\geq2$. Continuing this process $a-1$
steps, we end up with
\[
\begin{array}{lcl}
f_0(x)&=&\beta_{00}(x+1)^{i_{00}}\alpha_{00}(x)+\beta_{01}(x+1)^{i_{01}}\alpha_{01}(x)+\dots+\beta_{0(a-1)}(x+1)^{i_{0(a-1)}}\alpha_{0(a-1)}(x)\\
f_1(x)&=&\beta_{11}(x+1)^{i_{11}}\alpha_{11}(x)+\beta_{12}(x+1)^{i_{12}}\alpha_{12}(x)+\dots+\beta_{1(a-1)}(x+1)^{i_{1(a-1)}}\alpha_{1(a-1)}(x)\\
&\vdots&\\
f_{a-2}(x)&=&\beta_{(a-2)(a-2)}(x+1)^{i_{(a-2)(a-2)}}\alpha_{(a-2)(a-2)}(x)+\beta_{(a-2)(a-1)}(x+1)^{i_{(a-2)(a-1)}}\alpha_{(a-2)(a-1)}(x)\\
f_{a-1}(x)&=&\beta_{(a-1)(a-1)}(x+1)^{i_{(a-1)(a-1)}}\alpha_{(a-1)(a-1)}(x)\\
f_{a}(x)&=&0\\
f_{a+1}(x)&=&0\\
&\vdots&\\
f_{n}(x)&=&0\\
\end{array}
\]
Finally, we see $g_j\in\id{f_0,f_1,f_2,\dots,f_{a-1}}=I$.
\end{proof}

To illustrate the previous algorithmic proof, we provide the
following example.

\begin{example}
Let $f_1,f_2,f_3\in \frac{Z_{9}[x]}{x^{27}+1}$ s.t.
\begin{eqnarray*}
f_1(x)&=&(x+1)-3\\
f_2(x)&=&(x+1)^2+3(x+1)\\
f_3(x)&=&(x+1)^3+3(x+1)\\
\end{eqnarray*}
After applying the first step in the algorithm in the above proof we
have
\begin{eqnarray*}
f_2'(x)&=&f_2(x)-f_1(x)(x+1)=6(x+1)\\
f_3'(x)&=&f_2(x)-f_1(x)(x+1)^3=3(x+1)[1+(x+1)]\\
\end{eqnarray*}
Then one final step gives
\begin{eqnarray*}
f_3''(x)&=&f_3'(x)+f_2'(x)[1+(x+1)]=0\\
\end{eqnarray*}
So, $\id{f_1,f_2,f_3}=\id{f_1,f_2'}$.
\end{example}

We conclude this section by providing a method for finding the
Hamming distances of codes in $\R$ where $p$ is an odd prime and
$a>1$ is provided. We use the canonical definition of Hamming weight
for polynomial based codes i.e. given $\bar{c}(x) \in \R$ where we
consider $c(x)$ as the polynomial representative of degree less than
$p^s$ in the coset $\bar{c}(x)=c(x)+\id{x^{p^s}+1}$, the Hamming
weight of $\bar{c}(x)$, $w(\bar{c}(x))$, is the number of non-zero
coefficients of $c(x)$. The minimum distance $d$ is then defined in
the usual way.

\begin{remark}
\label{remark.isometry} It should be clear that the isomorphism in
Lemma \ref{lemma.iso} is an isometry when the Hamming weight in
$\frac{GR(p,m)[x]}{\id{x^{p^s}+1}}$ is defined similarly to the
weight in $\R$.
\end{remark}

Theorem 4.11 of \cite{dinh_2008} provides the distances for any code
in $\frac{GR(p,m)[x]}{\id{x^{p^s}+1}}$, which we include here.

\begin{lemma}[Theorems 4.11, \cite{dinh_2008}]
\label{lemma.hai2} In $\frac{GR(p,m)[x]}{\id{x^{p^s}+1}}$, for
$0\leq i\leq p^s$
\begin{equation*}
d(\id{(x+1)^i})=\left\{
\begin{array}{ll}
1&\textrm{ if } i=0,\\
\beta+2&\textrm{ if } \beta p^{s-1}+1\leq i\leq (\beta+1)p^{s-1} \textrm{ where } 0\leq \beta\leq p-2,\\
(t+1)p^k&\textrm{ if } p^s-p^{s-k}+(t-1)p^{s-k-1}+1\leq i\leq p^s-p^{s-k}+tp^{s-k-1},\\
&\textrm{ where } 1\leq t\leq p-1 \textrm{ and } 1\leq k\leq s-1, \\
0&\textrm{ if } i=p^s,\\
\end{array}
\right.
\end{equation*}
\end{lemma}

Using one additional result from \cite{salagean_2006} the distances
of all codes in $\R$ can be found. Moreover, Lemma \ref{lemma.ideal}
is algorithm based, so using all these results, an algorithm exists
for finding the distances of these codes.

\begin{lemma}[adapted from Theorem 6.1 in \cite{salagean_2006}]
\label{lemma.sal} Given the set a set a generators for a code
$I\vartriangleleft\R$ as in Lemma \ref{lemma.ideal},
$d(I)=d(\id{p^{a-1}f_r})$.
\end{lemma}

\begin{remark}
\label{remark.dist} Given Remark \ref{remark.isometry} and Lemmas
\ref{lemma.hai2} and \ref{lemma.sal}, the distance of any code in
$\R$ can be determined. Let $I\vartriangleleft\R$. We can find
$f_1,\dots,f_r\in\R$ such that
$I=\id{p^{j_0}f_0,\dots,p^{j_r}f_{r}}$ where this set satisfies the
properties of Lemma \ref{lemma.ideal}. Then Lemma \ref{lemma.sal}
implies that $d(I)=d(\id{p^{a-1}f_r})$. Next we view $f_r$ in
canonical form. Write
\[
f(x)=\beta_0(x+1)^{i_0}\alpha_0(x)+\beta_1p^{1}(x+1)^{i_1}\alpha_1(x)+\dots+\beta_{a-1}p^{{a-1}}(x+1)^{i_{a-1}}\alpha_{a-1}(x)
\]
where $\beta_k\in \Tm$ and $\alpha_k(x)\in\R$ is invertible. Note
since $f_r$ is monic, $\beta_0\neq 0$. So,
$p^{a-1}f_r=p^{a-1}\beta_0(x+1)^{i_0}\alpha_0(x)$. Since $\beta_0$
and $\alpha_0(x)$ are units,
$d(I)=d(\id{p^{a-1}f_r})=d(\id{p^{a-1}(x+1)^{i_0}})$. In light of
Remark \ref{remark.isometry}, the distance
$d(\id{p^{a-1}(x+1)^{i_0}})$ can be found using Lemma
\ref{lemma.hai2}.
\end{remark}

\section{Cyclic codes in $\frac{GR(p^a,m)[x]}{\id{x^{p^s}-1}}$ for arbitrary prime $p$}
\label{sect.struct2} Let us first consider the case when $p$ is an
odd prime. It is easy to see, arguing as in Proposition 5.1 in
\cite{dinh_2004}, that $\R\cong\frac{GR(p^a,m)[x]}{\id{x^{p^s}-1}}$
by sending $x$ to $-x$. Hence, all the results can in Section
\ref{sect.struct} translate easily into results about cyclic codes.
Let us therefore focus solely on the case when $p=2$ for the
remainder of this section.

Remember that in \cite{dinh_2005}, it was shown that
$\frac{GR(2^a,m)[x]}{\id{x^{2^s}+1}}$ is a chain ring. They also
computed the Hamming distance for most of the codes. Let us now
consider the cyclic case i.e. the code ambient $\T$. It turns out
that when $a>1$, this is not a chain ring but the structure is very
similar to the ring considered in Section \ref{sect.struct}. In this
section assume $a>1$ and as before $s>0$ since this produces the
trivial case. Most of the proofs here are very similar to their
analogs in Section \ref{sect.struct}. We include only the proofs
that need fundamental modification.

\begin{proposition}
\label{prop.nilp2} In $\T$, $(x+1)$ is nilpotent.
\end{proposition}

\begin{proof}
\[
\begin{array}{lll}
(x+1)^{2^s}     & = & x^{2^s}+\nck{2^s}{2^s-1}x^{2^{s}-1}+\dots+\nck{2^s}{1}x+1-1+1\\
                & = & x^{2^s}-1+2\alpha(x)\\
                & = & 2\alpha(x)\\
\end{array}
\]
where $\alpha(x)\in \T$. Then $(x+1)^{2^sa}=2^a(\alpha(x))^a=0$.
\end{proof}

The following propositions can be obtained from the parallel results
in Section \ref{sect.struct} by replacing $p$ with $2$ and
$x^{p^s}+1$ with $x^{2^s}-1$ and using Proposition \ref{prop.nilp2}
in lieu of Proposition \ref{prop.nilp} when needed.

\begin{proposition}
\label{prop.max2} The ambient ring $\T$ is local with radical
$J\left(\T\right)=\id{2,x+1}$.
\end{proposition}

\begin{proposition}
\label{prop.soc2} The socle $soc\left(\T\right)$ of $\T$ is the
simple module $\id{2^{a-1}(x+1)^{(2^s-1)}}$.
\end{proposition}

\begin{proposition}
\label{lemma.notchain2} In $\T$
\begin{enumerate}
\item $2\notin{\id{x+1}}$
\item $x+1\notin{\id{2}}$
\item $\T$ is not a chain ring
\item $\id{2,x+1}$ is not a principal ideal
\end{enumerate}
\end{proposition}

\begin{theorem}
The ambient ring $\T$ is a finite local ring with simple socle but
not a chain ring.
\end{theorem}
\begin{proof}
Result of Propositions \ref{prop.max2}, \ref{prop.soc2} and
\ref{lemma.notchain2}.
\end{proof}

The following Lemma is similar to Lemma \ref{lemma.main} with a
subtle difference in the divisor of $a_t(x)$ which is used in the
last line of the proof.

\begin{lemma}
\label{lemma.main2} In $\T$ for $t\geq 0$,
\[
(x+1)^{2^s+t2^{s-1}}=2^{t+1}b_t(x)(x+1)^{2^{s-1}}+a_t(x)
\]
where $b_t(x)$ is invertible and $2^{t+2}(x+1)|a_t(x)$.
\end{lemma}
\begin{proof} We proceed by induction on $t$. For $t=0$,
\[
\begin{array}{lcl}
0=x^{2^s}-1&=&((x+1)-1)^{2^s}-1\\
&=&(x+1)^{2^s}-\nck{2^s}{{2^s}-1}(x+1)^{{2^s}-1}+\nck{{2^s}}{{2^s}-2}(x+1)^{{2^s}-2}-\dots-\nck{{2^s}}{1}(x+1)\\
\end{array}
\]
By Lemma \ref{lemma.pd}
\[
\begin{array}{lcl}
(x+1)^{2^s}&=&\nck{2^s}{{2^s}-1}(x+1)^{{2^s}-1}-\nck{{2^s}}{{2^s}-2}(x+1)^{{2^s}-2}+\dots+\nck{{2^s}}{1}(x+1)\\
&=&\nck{2^s}{2^{s-1}}(x+1)^{2^{s-1}}+a_0(x)\\
&=&2b_0(x)(x+1)^{2^{s-1}}+a_0(x)\\
\end{array}
\]
for some $a_0(x)$ s.t. $2^2(x+1)|a_0(x)$ and
$b_0(x)=\frac{\nck{2^s}{2^{s-1}}}{2}$ which is invertible.

Now assume the result holds for $t-1$. So there exists some
$a_{t-1}(x)$ s.t. $2^{t+1}(x+1)|a_{t-1}(x)$ and $b_{t-1}(x)$
invertible where
$(x+1)^{2^s+(t-1)2^{s-1}}=2^tb_{t-1}(x)(x+1)^{2^{s-1}}+a_{t-1}(x)$.
So
\[
\begin{array}{lcl}
(x+1)^{2^s+t2^{s-1}}&=&(x+1)^{2^s+(t-1)2^{s-1}}(x+1)^{2^{s-1}}\\
&=&\left[2^tb_{t-1}(x)(x+1)^{2^{s-1}}+a_{t-1}(x)\right](x+1)^{2^{s-1}}\\
&=&2^tb_{t-1}(x)(x+1)^{2^s}+a_{t-1}(x)(x+1)^{2^{s-1}}\\
&=&2^tb_{t-1}(x)\left[2b_0(x)(x+1)^{2^{s-1}}+a_0(x)\right]+a_{t-1}(x)(x+1)^{2^{s-1}}\\
&=&2^{t+1}b_{t-1}(x)b_0(x)(x+1)^{2^{s-1}}+2^tb_{t-1}(x)a_0(x)+a_{t-1}(x)(x+1)^{2^{s-1}}\\
&=&2^{t+1}b_{t-1}(x)\left[b_0(x)+\frac{a_{t-1}(x)}{2^{t+1}}\right](x+1)^{2^{s-1}}+2^tb_{t-1}(x)a_0(x)\\
&=&2^{t+1}b_{t}(x)(x+1)^{2^{s-1}}+a_t(x)\\
\end{array}
\]
where $2^{t+2}(x+1)|a_{t}(x)$ and $b_{t}(x)$ invertible.
\end{proof}

\begin{corollary}
\label{cor.nil2} In $\T$, the nilpotency of $x+1$ is $(a+1)2^{s-1}$.
\end{corollary}
\begin{proof}
By Lemma \ref{lemma.main2},
\[
(x+1)^{2^s+(a-2)2^{s-1}}=2^{a-1}b(x)(x+1)^{2^{s-1}}+a(x)
\]
for some $b(x)$ is invertible and $a(x)$ s.t. $2^{a}|a(x)$. So,
$a(x)=0$ and
\[
(x+1)^{2^s+(a-2)2^{s-1}}=2^{a-1}b(x)(x+1)^{2^{s-1}}.
\]
So,
\[
(x+1)^{2^s+(a-2)2^{s-1}}(x+1)^{2^{s-1}-1}=2^{a-1}b(x)(x+1)^{2^s-1}
\]
meaning
\[
(x+1)^{2^s+(a-1)2^{s-1}-1}=2^{a-1}b(x)(x+1)^{2^s-1}\neq 0.
\]
Finally,
\[
(x+1)^{2^s+(a-1)2^{s-1}}=2^{a-1}b(x)(x+1)^{2^s}=0.
\]
Hence the nilpotency of $x+1$ is $2^s+(a-1)2^{s-1}=(a+1)2^{s-1}$.
\end{proof}

The two proofs for Lemma \ref{lemma.agen} can be adapted to this
setting with the same substitutions as before.

\begin{lemma}
\label{lemma.agen2} In $\T$, any ideal is at most $a$-generated.
\end{lemma}

As was the case with most of the structure results, the Hamming
distance results in Section \ref{sect.struct} can easily be adapted
to this setting. The main results needed are the Hamming Distances
for the codes in $\frac{GR(2,m)[x]}{\id{x^{2^s}-1}}$. These
distances again were obtained in \cite{dinh_2008}.

\begin{lemma}[Corollary 4.12, \cite{dinh_2008}]
\label{lemma.hai3} In $\frac{GR(2,m)[x]}{\id{x^{2^s}-1}}$, for
$0\leq i\leq p^s$
\begin{equation*}
d(\id{(x+1)^i})=\left\{
\begin{array}{ll}
1&\textrm{ if } i=0,\\
2&\textrm{ if } 1\leq i\leq p^{s-1} \textrm{ where } 0\leq \beta\leq p-2,\\
2^{k+1}&\textrm{ if } 2^s-2^{s-k}+1\leq i\leq p^s-p^{s-k}+2^{s-k-1},\\
&\textrm{ where } 1\leq k\leq s-1, \\
0&\textrm{ if } i=2^s,\\
\end{array}
\right.
\end{equation*}
\end{lemma}

\begin{remark}
Given Lemma \ref{lemma.hai3}, the same method as in Section
\ref{sect.struct} of Remark \ref{remark.dist} can be applied here to
compute the Hamming distances for codes in $\T$.
\end{remark}

\bibliography{../SteveSzaboRefs}
\bibliographystyle{is-abbrv}

\end{document}